\documentclass[aps,pra,twocolumn,superscriptaddress,nofootinbib,longbibliography]{revtex4-2}

\usepackage{amsmath,amssymb,amsthm,graphicx}
\usepackage{xcolor}
\usepackage{tikz}
\usetikzlibrary{arrows.meta,positioning,calc,fit}
\usepackage[colorlinks=true,linkcolor=blue,citecolor=blue,urlcolor=blue]{hyperref}

\newtheorem{definition}{Definition}
\newtheorem{example}{Example}
\newtheorem{theorem}{Theorem}

\newtheorem{proposition}{Proposition}

\newcommand{\ket}[1]{|#1\rangle}
\newcommand{\bra}[1]{\langle #1|}

\newcommand{\Id}{I}

\begin{document}

\title{Inclusion-Minimal local indistinguishability: a weak form of nonlocality}
\author{Mao-Sheng Li}

\affiliation{School of Mathematics, South China University of Technology, Guangzhou 510641, China}

\author{Zong-Xing Xiong}
 
\affiliation{Institute of Quantum Computing and Software, School of Computer
Science and Engineering, Sun Yat-Sen University, Guangzhou 510006, China
}

\author{Zhu-Jun Zheng}

\affiliation{School of Mathematics, South China University of Technology, Guangzhou 510641, China}

\author{Yan-Ling Wang}
\email{wangylmath@yahoo.com}
\affiliation{
School of Computer Science and Technology,
Dongguan University of Technology,
Dongguan, 523808, China
}

\date{\today}

\begin{abstract}
Local discrimination of quantum states is a fundamental task in distributed quantum information processing and underlies applications such as quantum communication, data hiding, and secret sharing. Here we investigate a weak form of local indistinguishability by asking how easily it can disappear when the candidate set is reduced or an additional copy of the unknown state is supplied. We introduce inclusion-minimal locally indistinguishable sets, namely, locally indistinguishable sets for which every proper subset is perfectly distinguishable by local operations and classical communication (LOCC), and show that every finite locally indistinguishable set contains such a subset. We further find that any inclusion-minimal locally indistinguishable sets becomes perfectly distinguishable by LOCC when two identical copies are available, although a single copy is insufficient. Fininally, We construct explicit inclusion-minimal locally indistinguishable product-state sets in $(\mathbb C^d)^{\otimes n}$ for every odd $d=2k+1$ and $n\ge2$. These results show that local indistinguishability can be nontrivial at the single-copy level yet fragile under either the removal of candidate states or a modest increase in copy resources, providing a complementary perspective on the structure of quantum nonlocality.
\end{abstract}

\maketitle


\section{Introduction}
Perfectly distinguishing distributed quantum states is a basic information processing task when the subsystems cannot be brought together.  It asks how much information encoded in a composite quantum system can be recovered using only local operations and classical communication (LOCC), and is directly relevant to quantum communication, data hiding, secret sharing, and other distributed protocols~\cite{Terhal2001,DiVincenzo2002,EggelingWerner2002,Horodecki2003,Matthews2009,Chitambar2014}.  The restriction to LOCC can change the problem qualitatively: mutually orthogonal states are perfectly distinguishable by a global measurement, but need not remain so under local measurements.  Bennett \emph{et al.} gave the first striking example, showing that a complete set of orthogonal product states can be locally indistinguishable~\cite{Bennett1999}.  Thus entanglement of the candidate states is not required for a separation between global and local access to classical information.  This phenomenon, known as nonlocality without entanglement, has become one of the standard settings for studying the limitations of LOCC.

Lots of work has since developed the local discrimination problem from several directions.  For entangled states, fundamental results include perfect LOCC discrimination of any two orthogonal pure states, the local distinguishability of Bell and generalized Bell states, and bounds for larger orthogonal families~\cite{Walgate2000,WalgateHardy2002,Ghosh2001,Fan2004,Nathanson2005,Watrous2005,Hayashi2006,Duan2007,Cohen2007,OwariHayashi2006,OwariHayashi2008,Bandyopadhyay2011,Yu2012,Hashimoto2021,Yang2018,Wang2025Detectors,Wang2025LU}.  For product states, many examples and constructions are now known in bipartite and multipartite systems, together with related results for unextendible product bases and for separable, PPT, finite round, and asymptotic local measurements~\cite{BennettUPB1999,DiVincenzoUPB2003,Wang2015,Zhang2014,Zhang2016,Cosentino2013,Childs2013,BandyopadhyayNathanson2013,BandyopadhyayEtAl2015,Kleinmann2011}.  Another line of research asks for stronger forms of local indistinguishability.  Local irreducibility, strong or genuine nonlocality, and local stability impose more restrictive conditions on the orthogonality preserving local measurements available to the parties, and have led to systematic constructions and cardinality bounds in general multipartite systems~\cite{Halder2019,ZhangZhang2019,Rout2019,WangLiYung2021,Rout2021,ShiUPB2022,ShiN2022,LiWang2023,CaoLiZuo2023,Zhou2023,Xiong2023,Hu2024,Bhunia2024,Zhen2024,He2024,Zhou2025,Zhen2025,Xiong2026Smallest}.  These results clarify how robust local indistinguishability can be against the first nontrivial local measurement or against regrouping of the parties.

The opposite question is equally natural: how little additional information or resource is needed to make a locally indistinguishable set distinguishable?  Minimal nonlocality addresses this question through state deletion.  A locally indistinguishable set is called minimally nonlocal if there exists at least one state whose removal makes the remaining candidates perfectly distinguishable by LOCC~\cite{Zhu2023Minimal}.  Recent work has realized this idea to multipartite systems and unequal local dimensions and has also studied how entanglement assistance, and related discrimination resources alter locally indistinguishable sets~\cite{Xu2025Minimal,Zhao2026Minimal,XuCaoXin2026,BandyopadhyayRusso2024,SrivastavHalder2025,Xu2026Five,Murshid2026}.  This viewpoint is different from strong nonlocality or local stability: rather than restricting which local measurement can start a protocol, it asks how sensitive the entire discrimination problem is to a change in the candidate set or in the available resources.  The present definition of minimal nonlocality, however, requires only one favorable deletion.  Other states in the same set may be dispensable, and the set may still contain a smaller locally indistinguishable subset.

This motivates the set-inclusion question considered here: can one identify locally indistinguishable sets for which \emph{every} state is essential?  We call such a set \emph{inclusion-minimal locally indistinguishable}: every proper subset is perfectly distinguishable by LOCC, or equivalently for finite sets, deletion of any single state restores local distinguishability.  We first show that every finite locally indistinguishable set contains an inclusion-minimal locally indistinguishable subset.  Hence every finite example contains a locally indistinguishable subset that is minimal under set inclusion.  We then establish an operational consequence that does not follow from the definition alone.  For any minimally nonlocal set of mutually orthogonal pure states, two identical copies of the unknown state are always sufficient for perfect LOCC discrimination, whereas one copy is not.  In particular, every inclusion-minimal locally indistinguishable pure state set loses its local indistinguishability as soon as one additional copy is supplied.  Finally, we construct explicit inclusion-minimal product state sets.  A twelve-state set in $(\mathbb C^3)^{\otimes3}$ gives the basic example, and the construction extends to $(\mathbb C^d)^{\otimes n}$ for every odd $d=2k+1$ and $n\geq2$, with $4nk$ states.   

\section{Basic notions and definitions}

Let
 $$
\mathcal H=\bigotimes_{r=1}^n \mathcal H_{A_r}
 $$
be an  $n $-partite Hilbert space.  A set of mutually orthogonal pure states
 $$
\mathcal S=\{\ket{\psi_i}\}\subset\mathcal H
 $$
is said to be perfectly distinguishable by LOCC if there exists an LOCC protocol that identifies the
given state with unit probability.  Otherwise,  $\mathcal S $ is said to be locally
indistinguishable.

A positive operator-valued measure (POVM) on a local system  $\mathcal H_A $ is a collection
 $\{E_x\}_{x\in X} $ of positive semidefinite operators such that
 $$
\sum_{x\in X}E_x=\Id_A .
 $$
The POVM is said to be trivial if every POVM element is proportional to the identity operator.

\begin{definition}[Orthogonality-preserving local measurement, \cite{WalgateHardy2002}]
Let  $\mathcal S=\{\ket{\psi_i}\} $ be an orthogonal set of multipartite pure states.
A local POVM  $\{E_x=M_x^\dagger M_x\}_{x\in X} $ performed by one party  $A_i$ is called
orthogonality preserving with respect to  $\mathcal S $ if, for every  $x\in X $,
 $$
\bra{\psi_i}(E_x\otimes I_{\overline{A_i}})\ket{\psi_j}=0,\qquad i\neq j  
 $$
where $\overline{A_i}:=\{A_1,\cdots,A_n\}\setminus\{A_i\}.$
\end{definition}

We use the fact that in any perfect LOCC
discrimination protocol, some party has to start with a nontrivial orthogonality-preserving local
measurement.  If no such measurement exists, then no LOCC protocol can perfectly distinguish the set.

\begin{definition}[Minimal nonlocality, \cite{Zhu2023Minimal}] A locally indistinguishable set  $\mathcal S $ is said to be
minimally nonlocal if there exists at least one state  $\ket{\psi}\in\mathcal S $ such that
 $$
\mathcal S\setminus\{\ket{\psi}\}
 $$
is perfectly distinguishable by LOCC.
\end{definition}

This notion describes a deletion-sensitive form of local indistinguishability: after excluding a
suitable candidate state, the remaining set can be perfectly distinguished by LOCC.  However, the
condition only refers to one possible deletion.  It does not require that every state in
 $\mathcal S $ plays the same role.  In particular, a minimally nonlocal set may still contain
smaller locally indistinguishable subsets.

Motivated by this point, we introduce the following stronger notion.  It captures the case where
each state is necessary for keeping the whole set locally indistinguishable: once any state is
removed, the remaining set becomes locally distinguishable.

\begin{definition}[Inclusion-minimal locally indistinguishable set]
A locally indistinguishable set  $\mathcal S $ of orthogonal quantum states is called
inclusion-minimal locally indistinguishable if every proper subset
 $$
\mathcal T\subsetneq \mathcal S
 $$
is perfectly distinguishable by LOCC.
\end{definition}

For finite sets, this is equivalent to the one-state-deletion condition:  $\mathcal S $ is
inclusion-minimal locally indistinguishable if and only if it is locally indistinguishable and, for
every  $\ket{\psi}\in\mathcal S $, the set
 $$
\mathcal S\setminus\{\ket{\psi}\}
 $$
is perfectly distinguishable by LOCC.  Indeed, every proper subset is contained in a
one-state-deleted subset, and an LOCC protocol that distinguishes a larger set also distinguishes any
smaller subset.

The terminology \emph{inclusion-minimal} specifies that the minimality is with respect to set
inclusion, not with respect to cardinality, dimension, entanglement, or the number of copies used for
discrimination.  In particular, an inclusion-minimal locally indistinguishable set is automatically
minimally nonlocal in the above sense, while the converse need not hold.


\section{Structural results and exact two-copy LOCC discrimination}

We now move from definitions to general facts.  The first point is that the inclusion-minimal notion
is not tied to a special construction.  Whenever a finite locally indistinguishable set exists, one can
find such a subset inside it.

\begin{theorem}
\label{thm:IMLI-core}
Every finite locally indistinguishable set contains at least one inclusion-minimal locally
indistinguishable subset.  As a consequence, every finite locally indistinguishable set contains a
minimally nonlocal subset.
\end{theorem}

\begin{proof}
Let  $\mathcal S $ be a finite locally indistinguishable set.  Define
 $$
\mathrm{NL}(\mathcal S)
=
\{\mathcal T\subseteq\mathcal S:\mathcal T\ \text{is locally indistinguishable}\}.
 $$
Since  $\mathcal S\in\mathrm{NL}(\mathcal S) $, this family is nonempty.  Choose
 $\mathcal T\in\mathrm{NL}(\mathcal S) $ with the smallest possible cardinality.

By construction,  $\mathcal T $ is locally indistinguishable.  For every
 $\ket{\psi}\in\mathcal T $, the one-state-deleted set
 $$
\mathcal T\setminus\{\ket{\psi}\}
 $$
must be locally distinguishable; otherwise it would be a locally indistinguishable subset of
 $\mathcal S $ with cardinality strictly smaller than  $|\mathcal T| $, contradicting the choice of
 $\mathcal T $.  Hence every one-state-deleted subset of  $\mathcal T $ is LOCC distinguishable.  By
the finite-set equivalence in the definition above, every proper subset of  $\mathcal T $ is LOCC
distinguishable.  Therefore  $\mathcal T $ is inclusion-minimal locally indistinguishable.

Since every inclusion-minimal locally indistinguishable set is minimally nonlocal, every finite
locally indistinguishable set also contains a minimally nonlocal subset.
\end{proof}

The theorem gives the set-inclusion part of the story.  It says that, after passing to a suitable
subset, every state becomes necessary in the following operational sense: if any one candidate is
excluded, the remaining candidates can be distinguished by LOCC.  Thus inclusion-minimal local
indistinguishability is weak from the viewpoint of side information.

We next look at weakness from a different viewpoint.  Instead of giving the parties side information,
we give them another copy of the unknown state.  A locally indistinguishable set whose LOCC copy
number is two is also weak in this sense: one copy is not enough, but the next possible number of
copies already suffices.  The next theorem shows that these two viewpoints are linked for orthogonal
pure states.

For an orthogonal set  $\mathcal S $, let  $N_{\mathrm{LOCC}}(\mathcal S) $ denote the smallest
number of identical copies of the unknown state that is sufficient for perfect LOCC discrimination.
For every finite orthogonal set  $\mathcal S $, one has  $N_{\mathrm{LOCC}}(\mathcal S)\le
|\mathcal S|-1 $, because  $ |\mathcal S|-1 $ copies can be used to exclude candidates one by one.

\begin{theorem}
\label{thm:two-copy}
Let  $\mathcal S $ be an inclusion-minimal locally indistinguishable set of orthogonal pure states.
Then
 $$
N_{\mathrm{LOCC}}(\mathcal S)=2 .
 $$
\end{theorem}

\begin{proof}
Since  $\mathcal S $ is locally indistinguishable by definition, one copy is not sufficient.  It
remains only to show that two copies always suffice.

Let
 $$
\mathcal S=\{\ket{\psi_1},\ldots,\ket{\psi_m}\}.
 $$
Since any two orthogonal pure states are perfectly distinguishable by LOCC~\cite{Walgate2000}, we
must have  $m\ge3 $.

Use the first copy only to exclude one candidate.  The parties choose two states, say
 $\ket{\psi_1} $ and  $\ket{\psi_2} $, and run a two-state LOCC discrimination protocol on the first
copy.  Since this protocol distinguishes  $\ket{\psi_1} $ from  $\ket{\psi_2} $ perfectly, each
terminal outcome excludes at least one of these two states with certainty.  Hence, after observing
the outcome, the parties know that the unknown state belongs to
 $$
\mathcal S\setminus\{\ket{\psi_\ell}\}
 $$
for some known  $\ell\in\{1,2\} $.

Because  $\mathcal S $ is inclusion-minimal locally indistinguishable, the reduced set
 $\mathcal S\setminus\{\ket{\psi_\ell}\} $ is perfectly LOCC distinguishable.  The parties therefore
apply the corresponding one-copy LOCC protocol to the second copy.  This identifies the state with
certainty, and hence two copies suffice.  Combining this with one-copy local indistinguishability
gives  $N_{\mathrm{LOCC}}(\mathcal S)=2 $.
\end{proof}

It is natural to ask whether a minimally nonlocal set may require more than one copy for perfect local discrimination. We show that this is never the case: two copies are always sufficient.

\begin{theorem}
	\label{thm:two-copy}
	Let  $\mathcal S $ be a minimally nonlocal set of mutually orthogonal pure states. Then
	 $$
	N_{\mathrm{LOCC}}(\mathcal S)=2 .
	 $$
\end{theorem}

\begin{proof}
	By definition,  $\mathcal S $ is locally indistinguishable. Hence one copy is not sufficient, and therefore
	 $$
	N_{\mathrm{LOCC}}(\mathcal S)\geq 2 .
	 $$
	It remains to prove that two copies suffice.	
	Write
	 $$
	\mathcal S=\{\ket{\psi_1},\ldots,\ket{\psi_m}\}.
	 $$
	Since  $\mathcal S $ is minimally nonlocal, there exists an index
	 $\ell\in\{1,\ldots,m\} $ such that the subset
	$
	\mathcal S\setminus\{\ket{\psi_\ell}\}
	$
	is perfectly distinguishable by LOCC. Let  $\mathcal P $ be an LOCC protocol that perfectly distinguishes this subset.
	
	Now suppose that two copies of an unknown state from  $\mathcal S $ are provided. On the first copy, the parties apply the protocol  $\mathcal P $. If the unknown state belongs to  $\mathcal S\setminus\{\ket{\psi_\ell}\} $, then the outcome of  $\mathcal P $ identifies it uniquely. More generally, for each outcome of  $\mathcal P $ that identifies some state  $\ket{\psi_i} $, where  $i\neq \ell $, the only possible states consistent with that outcome are  $\ket{\psi_i} $ and  $\ket{\psi_\ell} $. Indeed, such outcome of the protocol  $\mathcal P $ perfectly excludes all   states in  $\mathcal S\setminus\{\ket{\psi_i},\ket{\psi_\ell}\} $.
	
	It remains only to distinguish between the two orthogonal pure states  $\ket{\psi_i} $ and  $\ket{\psi_\ell} $. This can always be done perfectly by LOCC using the second copy, by the standard theorem that any two orthogonal pure states are perfectly distinguishable by LOCC. Hence two copies are sufficient, and so
	$
	N_{\mathrm{LOCC}}(\mathcal S)\leq 2 .
	$
	Combining this with  $N_{\mathrm{LOCC}}(\mathcal S)\geq 2 $, we obtain
	$
	N_{\mathrm{LOCC}}(\mathcal S)=2.
$
\end{proof}

\begin{figure}[t]
	\centering
 
		\includegraphics[width=0.45\textwidth]{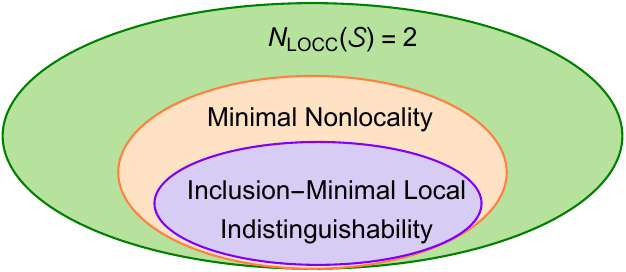}
		\caption{Relations among minimally nonlocal sets, inclusion-minimal locally indistinguishable
			sets, and sets with LOCC copy number  $N_{\mathrm{LOCC}}(\mathcal S)=2 $.}\label{fig:relations}		 
\end{figure}
The relations among the three classes discussed above are summarized in
Fig.~\ref{fig:relations}.


\section{Product-state constructions}
The preceding theorems explain why inclusion-minimal locally indistinguishable sets have exact LOCC
copy number two.  They do not yet show how such sets look.  We now build concrete examples in the
most transparent way: first a three-qutrit model, then the general odd-dimensional multipartite
family.

\begin{figure*}[t]
\centering

\begin{minipage}[t]{0.32\textwidth}
    \centering
    \includegraphics[width=\textwidth]{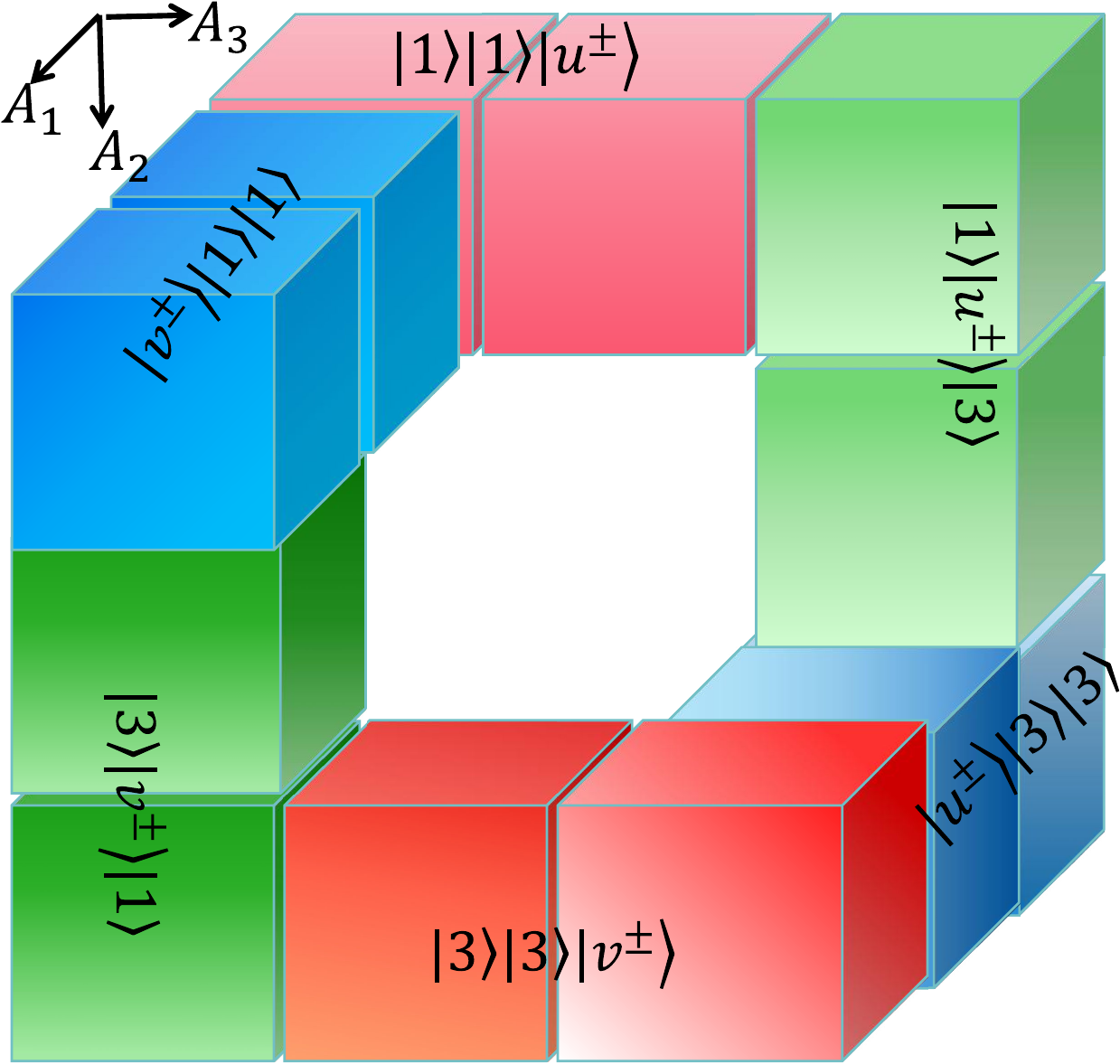}
    \vspace{1mm}
    
    {\small \textbf{(a)}  }
\end{minipage}
\hfill
\begin{minipage}[t]{0.31\textwidth}
    \centering
    \includegraphics[width=\textwidth]{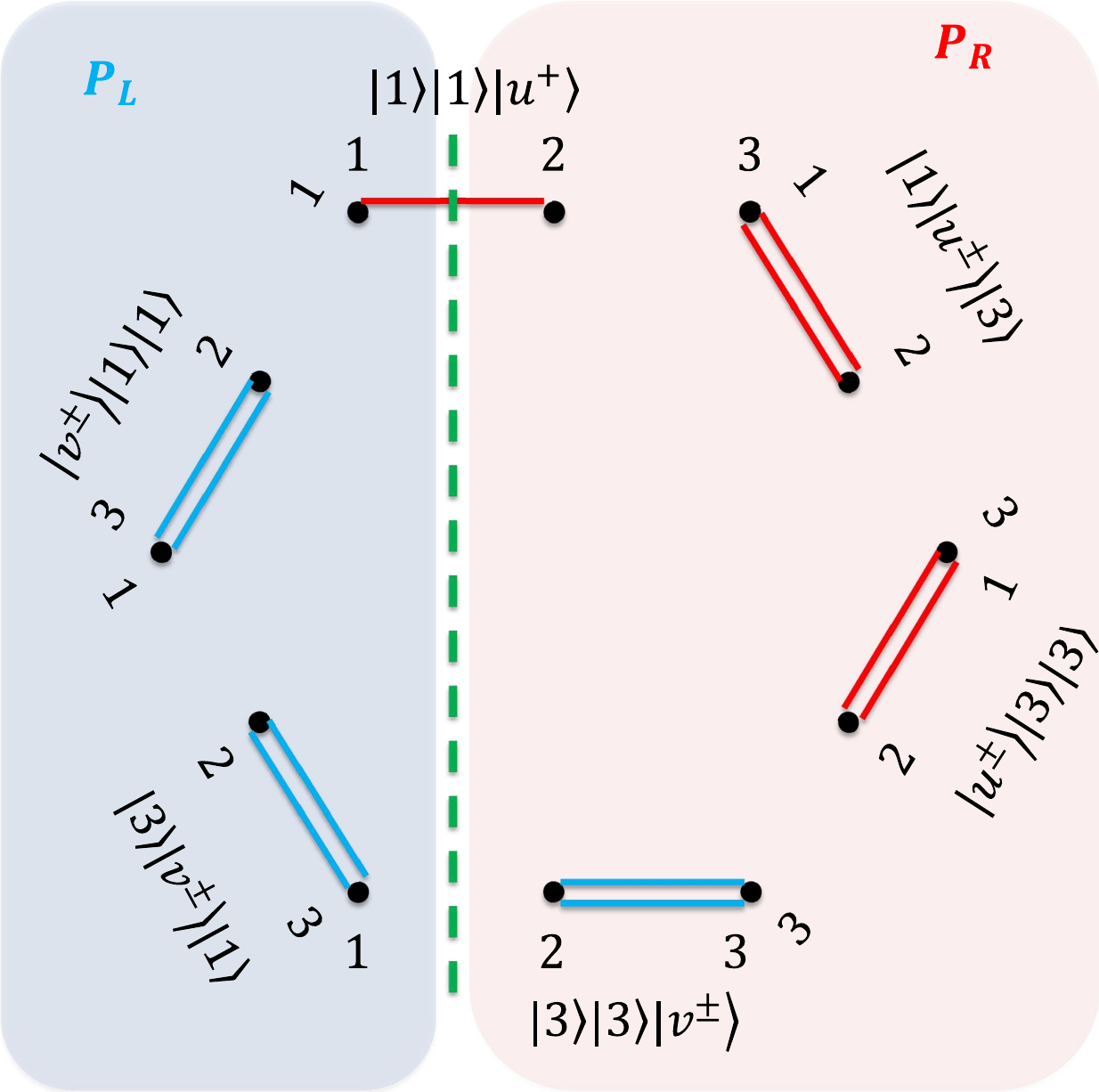}
    \vspace{1mm}
    
    {\small \textbf{(b)}  }
\end{minipage}
\hfill
\begin{minipage}[t]{0.14\textwidth}
    \centering
    \includegraphics[width=\textwidth]{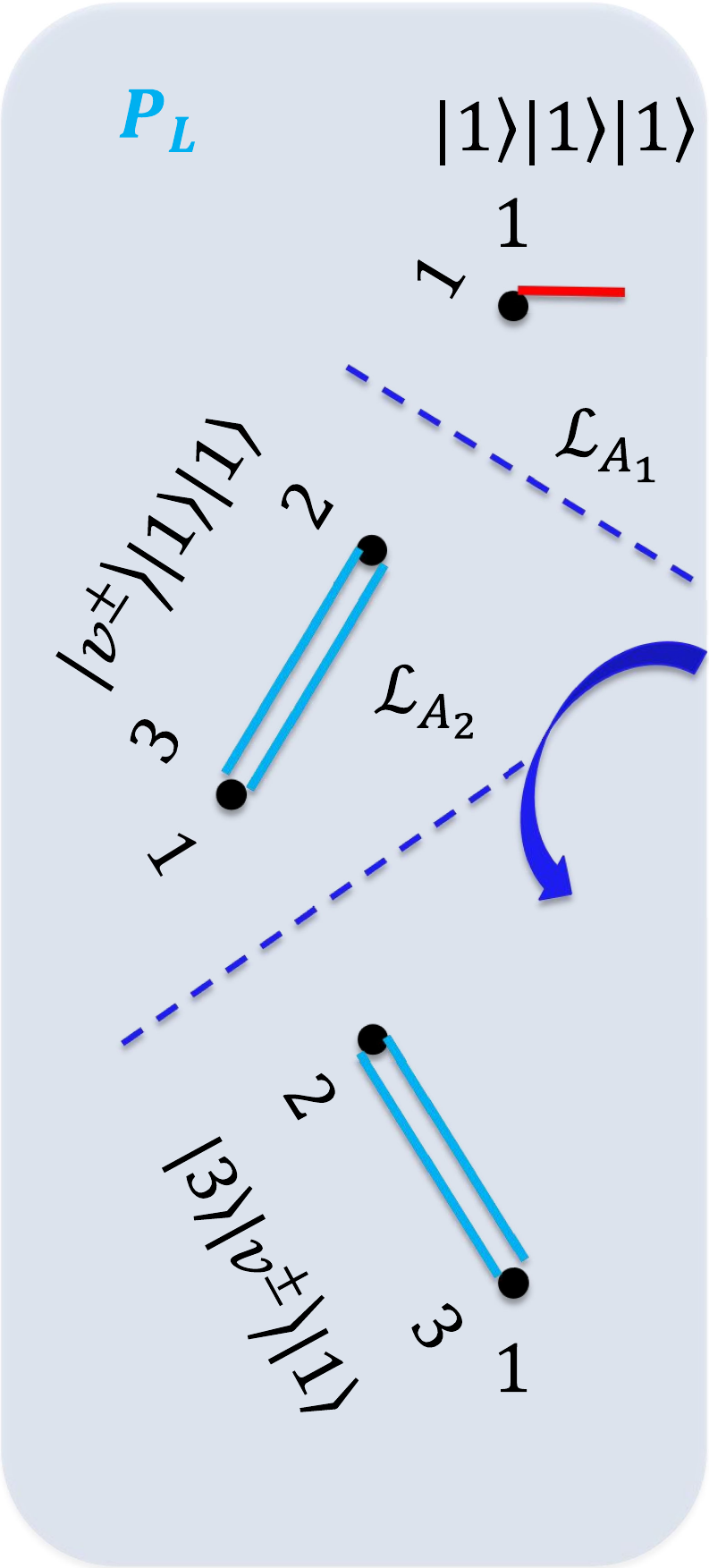}
    \vspace{1mm}
    
    {\small \textbf{(c)} }
\end{minipage}
\hfill
\begin{minipage}[t]{0.19\textwidth}
    \centering
    \includegraphics[width=\textwidth]{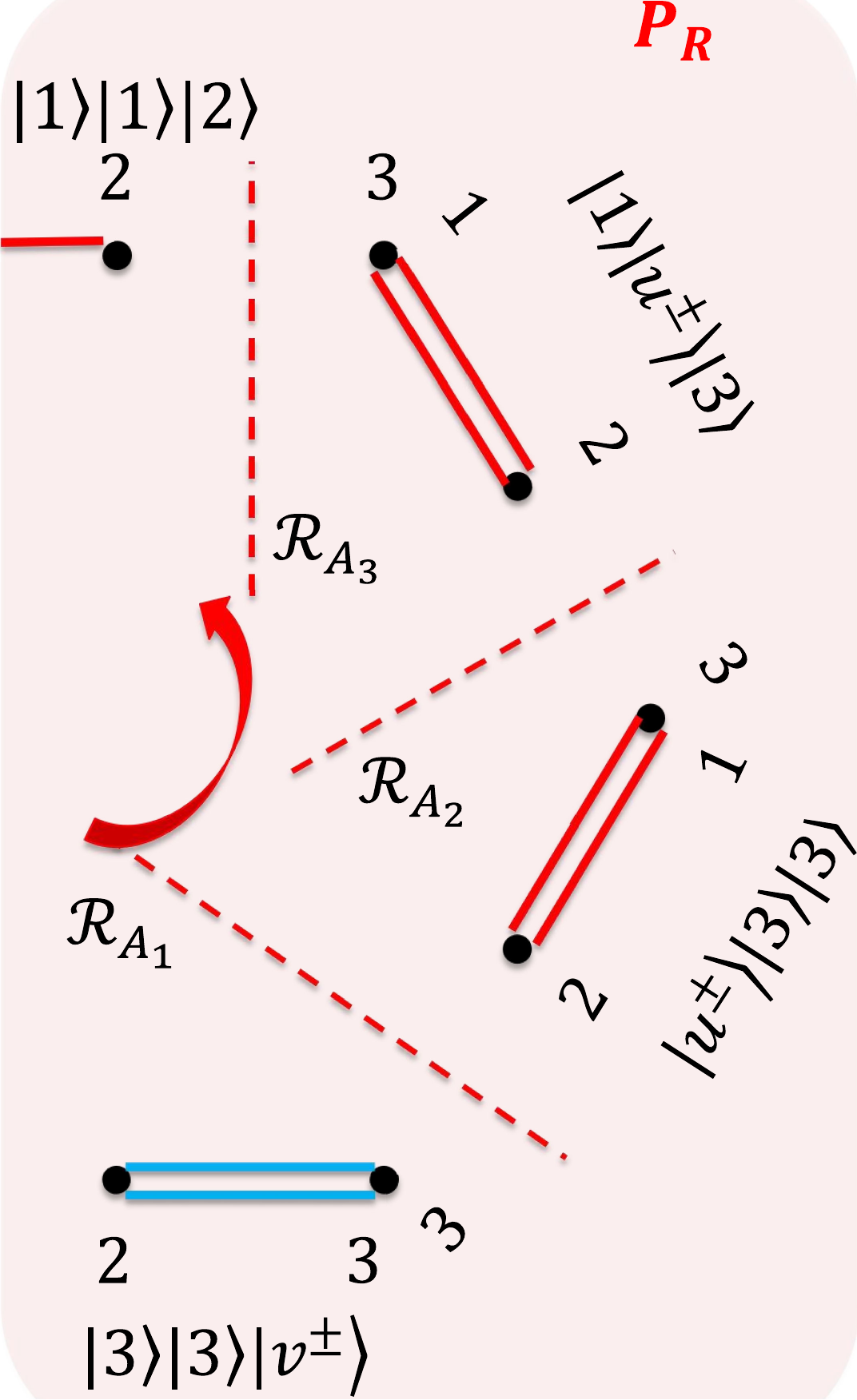}
    \vspace{1mm}
    
    {\small \textbf{(d)} }
\end{minipage}

\caption{
(Color online)  
(a)  The geometry of the constructed states in $\mathcal{S}_3$ defined in Eq. \eqref{eq:three-qutrit-example}.
(b) The intuition of the connectivity of the states in $\mathcal{S}_3$  and the effect of the first projection play on $A_3$.
(c) The sequential measurements on the $P_L$ branch.  (d) The sequential measurements on the $P_R$ branch.
}
\label{fig:big-figure33}
\end{figure*}

\begin{example}
Let  $\ket{1},\ket{2},\ket{3} $ be the computational basis of  $\mathbb C^3 $, and write
	 $$
	\ket{u^\pm}=\ket1\pm\ket2,
	\qquad
	\ket{v^\pm}=\ket2\pm\ket3 .
	 $$
	Consider the following twelve orthogonal product states in
	 $\mathbb C^3\otimes\mathbb C^3\otimes\mathbb C^3 $:

	\begin{align}
		\mathcal S_{3,3}=\{&
		\ket1\ket1\ket{u^\pm},
		\ket1\ket{u^\pm}\ket3,
		\ket{u^\pm}\ket3\ket3, \nonumber\\
		&\ket3\ket3\ket{v^\pm},
		\ket3\ket{v^\pm}\ket1,
		\ket{v^\pm}\ket1\ket1
		\}.                                      \label{eq:three-qutrit-example}
	\end{align}	
The set  $\mathcal S_{3,3} $ in Eq.~\eqref{eq:three-qutrit-example} is inclusion-minimal locally
indistinguishable.
\end{example}

\begin{proof} The geometric structure of  $\mathcal S_{3,3} $ is illustrated in
	Fig.~\ref{fig:big-figure33}(a), where the twelve states form a closed six-edge chain.
We first prove that the complete set is locally indistinguishable.  By cyclic symmetry, it is enough to consider
an orthogonality-preserving POVM element  $E=(a_{pq})_{p,q=1}^3 $ on party  $A_1 $.  The same argument
then applies to  $A_2 $ and  $A_3 $.

We show first that all off-diagonal entries of  $E $ vanish.  Compare the states
 $\ket{v^\eta}\ket1\ket1 $ and  $\ket1\ket1\ket{u^\epsilon} $, where
 $\epsilon,\eta\in\{+,-\} $.  The local factors on  $A_2 $ have nonzero overlap, and the local factors
on  $A_3 $ also have nonzero overlap.  Hence orthogonality preservation gives
 $$
\bra{v^\eta}E\ket1=0,
\qquad \eta\in\{+,-\}.
 $$
Taking the two signs separately gives
 $$
a_{21}=a_{31}=0 .
 $$
Next compare  $\ket{u^\epsilon}\ket3\ket3 $ with  $\ket3\ket3\ket{v^\eta} $.  Again the local factors
on the other two parties have nonzero overlap, and hence
 $$
\bra{u^\epsilon}E\ket3=0,
\qquad \epsilon\in\{+,-\} .
 $$
Thus
 $$
a_{13}=a_{23}=0 .
 $$
Together with the Hermiticity of  $E $, these equations imply that every off-diagonal entry of  $E $
is zero.

It remains to compare the diagonal entries.  Since  $E $ is diagonal, the orthogonality of the pair
 $\ket{u^+}\ket3\ket3 $,  $\ket{u^-}\ket3\ket3 $ gives
 $$
0=\bra{u^+}E\ket{u^-}=a_{11}-a_{22}.
 $$
Similarly, the orthogonality of the pair  $\ket{v^+}\ket1\ket1 $,  $\ket{v^-}\ket1\ket1 $ gives
 $$
0=\bra{v^+}E\ket{v^-}=a_{22}-a_{33}.
 $$
Therefore
 $$
a_{11}=a_{22}=a_{33},
 $$
and hence  $E\propto I_3 $.  Thus party  $A_1 $ cannot perform a nontrivial
orthogonality-preserving measurement.  The same cyclic argument applies to  $A_2 $ and  $A_3 $.  The
set  $\mathcal S_{3,3} $ is   locally indistinguishable.

We now show that deleting any one state makes the remaining set locally distinguishable.  It suffices
to give the protocol for one representative deletion; all other deletions follow from cyclic relabeling
of the parties, the basis reversal  $\ket q\mapsto\ket{4-q} $, and a phase flip on the corresponding
two-dimensional block.

Suppose that the state  $\ket1\ket1\ket{u^-} $ is deleted.  Party  $A_3 $ first performs
 $$
P_L=\ket1\bra1,
\qquad
P_R=\ket2\bra2+\ket3\bra3 .
 $$
The only state split by this measurement is the remaining partner  $\ket1\ket1\ket{u^+} $ (see in Fig. \ref{fig:big-figure33} (b)).  It gives
 $$
\ket{\gamma_L}=\ket1\ket1\ket1
\quad\text{on the }P_L\text{ branch},
 $$
and
 $$
\ket{\gamma_R}=\ket1\ket1\ket2
\quad\text{on the }P_R\text{ branch}.
 $$

On the  $P_L $ branch (see in Fig. \ref{fig:big-figure33} (c)), the possible states are
 $$
\ket{\gamma_L},
\qquad
\ket3\ket{v^\pm}\ket1,
\qquad
\ket{v^\pm}\ket1\ket1 .
 $$
Party  $A_1 $ first tests  $\ket1 $, i.e., the measurement $\mathcal{L}_{A_1}:=\{\pi_y=|1\rangle\langle 1|, \pi_n=I_3-\pi_y  \}$.  If the answer is `y', the state is  $\ket{\gamma_L} $.  If the
answer is `n', party  $A_2 $ tests  $\ket1 $,  i.e., the measurement $\mathcal{L}_{A_2}:=\{\pi_y=|1\rangle\langle 1|, \pi_n=I_3-\pi_y  \}$.  The `y' outcome leaves the pair
 $\ket{v^\pm}\ket1\ket1 $, which party  $A_1 $ distinguishes in a basis containing
 $\ket{v^+},\ket{v^-} $.  The `n' outcome leaves the pair  $\ket3\ket{v^\pm}\ket1 $, which party  $A_2 $
distinguishes in a basis containing  $\ket{v^+},\ket{v^-} $.  Hence the  $P_L $ branch is perfectly
LOCC distinguishable.

On the  $P_R $ branch (see in Fig. \ref{fig:big-figure33} (d)), the possible states are
 $$
\ket{\gamma_R},
\qquad
\ket1\ket{u^\pm}\ket3,
\qquad
\ket{u^\pm}\ket3\ket3,
\qquad
\ket3\ket3\ket{v^\pm}.
 $$
Party  $A_1 $ first tests  $\ket3 $, i.e.,  $\mathcal{R}_{A_1}:=\{\pi_y=|3\rangle\langle 3|, \pi_n=I_3-\pi_y  \}$.  If the answer is `y', the remaining pair is
 $\ket3\ket3\ket{v^\pm} $, which party  $A_3 $ distinguishes.  If the answer is `n', party  $A_2 $
tests  $\ket3 $, i.e.,  $\mathcal{R}_{A_2}:=\{\pi_y=|3\rangle\langle 3|, \pi_n=I_3-\pi_y  \}$.  The `y' outcome leaves the pair  $\ket{u^\pm}\ket3\ket3 $, which party  $A_1 $
distinguishes.  If the answer is `n', the remaining states are
 $$
\ket{\gamma_R}
\quad\text{and}\quad
\ket1\ket{u^\pm}\ket3 .
 $$
Then party  $A_3 $   tests  $\ket3 $, i.e.,  $\mathcal{R}_{A_3}:=\{\pi_y=|3\rangle\langle 3|, \pi_n=I_3-\pi_y  \}$. The `y' outcome leaves the pair  $\ket1\ket{u^\pm}\ket3 $, which party  $A_2 $
distinguishes.  At last, the `n' outcome identifies  $\ket{\gamma_R} $.

Combining the two branches, the set
 $\mathcal S_{3,3}\setminus\{\ket1\ket1\ket{u^-}\} $ is perfectly distinguishable by LOCC.  By the
symmetries mentioned above, deleting any one state from  $\mathcal S_{3,3} $ gives an LOCC
distinguishable set.  Therefore  $\mathcal S_{3,3} $ is inclusion-minimal locally indistinguishable.
\end{proof}

The proof of the three-qutrit example contains the whole mechanism.  The complete set is a closed
chain, while a one-state deletion opens the chain into two branches.  In the next proposition, the
same mechanism is implemented on a regular  $2n $-gon. 

\begin{proposition}\label{prop:odd-main}
Let  $d=2k+1 $ be an odd integer and  $n\geq 2 $. For  $r=1,\ldots,n $, $i=1,2,\cdots,k,$  define
\begin{align*}
\ket{\psi_{r,i}^{\pm}}
=
\ket{1}^{\otimes(r-1)}
\ket{u_i^\pm}_{A_r}
\ket{d}^{\otimes(n-r)},\\
 \ket{\phi_{r,i}^{\pm}}
=
\ket{d}^{\otimes(r-1)}
\ket{v_i^\pm}_{A_r}
\ket{1}^{\otimes(n-r)},
\end{align*} 
where 
$
\ket{u_i^\pm}:=\ket{2i-1}\pm\ket{2i}, \  \ket{v_i^\pm}:=\ket{2i}\pm\ket{2i+1}.$ 
Let
\begin{equation}\label{eq:Sn}
	\mathcal S_{n,d}
	=
	\left\{
	\ket{\psi_{r,i}^{\pm}},\ket{\phi_{r,i}^{\pm}}
	:
	r=1,\ldots,n,\ i=1,\ldots,k
	\right\}.
\end{equation} 
Then  $\mathcal S _{n,d} $ is an inclusion-minimal locally indistinguishable set of  $4nk $
orthogonal product states in  $(\mathbb C^d)^{\otimes n} $.
\end{proposition}

The proof follows the same idea as the preceding example; the full details are given in Appendix~\ref{Append}.
Moreover, the construction can be extended in a straightforward way to general multipartite systems
$
\bigotimes_{j=1}^n \mathbb C^{d_j},
$
where all local dimensions  $d_j $ are odd, but not necessarily equal.

\section{Operational meaning}

The results above can be summarized as a simple operational story.  A referee prepares one state from
an inclusion-minimal locally indistinguishable set  $\mathcal S $ and distributes its subsystems to
separated parties.  With one copy and no further information, the parties cannot identify the state
perfectly by LOCC.  If the referee announces that the state is not one fixed candidate
 $\ket{\psi_j} $, then the possible set is reduced to
 $\mathcal S\setminus\{\ket{\psi_j}\} $, which is locally distinguishable by definition.  Thus one
piece of exclusion information is enough to make the task locally solvable.

This phenomenon is different from ordinary data hiding based on entanglement or mixed states.  In our
explicit construction, all states are product states.  The difficulty comes from the global structure
of the ensemble and from the restriction to LOCC.  Hence inclusion-minimal locally indistinguishable
sets may be regarded as minimal ensembles in which the classical label is hidden from separated
parties but becomes locally accessible after any one candidate has been excluded.

The copy-number viewpoint gives another weak sense.  The original set cannot be distinguished by LOCC
with one copy, but any inclusion-minimal locally indistinguishable set of orthogonal pure states can
be distinguished by LOCC with two copies.  Thus the same class of sets is weak in two operational
senses: one piece of exclusion information is enough, and one additional copy is enough.

This observation also gives a possible primitive for distributed verification.  A verifier may encode
a label into an inclusion-minimal locally indistinguishable product-state set and distribute the
subsystems to several nodes.  Without exclusion information, the nodes cannot recover the label by
LOCC.  With the correct exclusion information, the label can be verified locally.  Of course, this
does not provide unconditional cryptographic security, since parties who are allowed to perform global
measurements can distinguish the original orthogonal set.  The value of the construction is structural:
it identifies a minimal set-theoretic form of local indistinguishability and shows how this property
changes after any one-state deletion.


\section{Conclusion and outlook}
We have introduced inclusion-minimal local indistinguishability, in which every state is essential for one-copy LOCC indistinguishability: deleting any candidate makes the remaining set perfectly distinguishable.  Every finite locally indistinguishable set contains such a subset.  Thus inclusion-minimal sets occur naturally as minimal subsets of finite locally indistinguishable sets rather than only through special constructions.

A second result relates this deletion property to copy-assisted discrimination.  If an orthogonal pure-state set becomes locally distinguishable after at least one suitable state is removed, then two identical copies of the unknown state always suffice for perfect LOCC discrimination, although one copy does not.  In this sense minimal nonlocality is fragile in two complementary ways.  It can be removed by suitable exclusion information, and it disappears once one additional copy of the unknown state is available.  This relation is independent of the particular product-state constructions used later in the paper.

We also constructed inclusion-minimal locally indistinguishable sets consisting entirely of product states.  The three-qutrit example contains twelve states, and the general family contains $4nk$ states in $(\mathbb C^d)^{\otimes n}$ for every odd $d=2k+1$ and $n\geq2$.  Before deletion these sets are locally stable, whereas every one-state-deleted subset admits a perfect LOCC protocol.  This separates two aspects of local discrimination: local stability constrains orthogonality-preserving local measurements for the full set, while inclusion minimality describes how the discrimination problem changes when the candidate set is reduced.

Several questions remain.  It would be useful to determine the minimum cardinality of an inclusion-minimal locally indistinguishable set in a fixed multipartite Hilbert space and to compare it with the corresponding bounds for locally stable and strongly nonlocal sets~\cite{LiWang2023,Zhen2024,Xiong2026Smallest}.  Even local dimensions are not covered by the present construction.  One may also ask for analogous minimality notions under finite-round or asymptotic LOCC, separable measurements, and PPT measurements, and for a quantitative comparison between state deletion, additional copies, and shared entanglement as resources for restoring local distinguishability~\cite{XuCaoXin2026,SrivastavHalder2025}.

\acknowledgments 
This work is supported by the Guangdong Basic and Applied Basic Research Foundation under Grants Nos. 2024A1515030023 and 2024A1515010380, the National Natural Science Foundation of China (Grant Nos. 12371458,    62401638  and 92465202).

\appendix

\section{\bf Proof of Proposition \ref{prop:odd-main}}\label{Append}

\begin{figure*}[t]
\centering

\begin{minipage}[t]{0.37\textwidth}
    \centering
    \includegraphics[width=\textwidth]{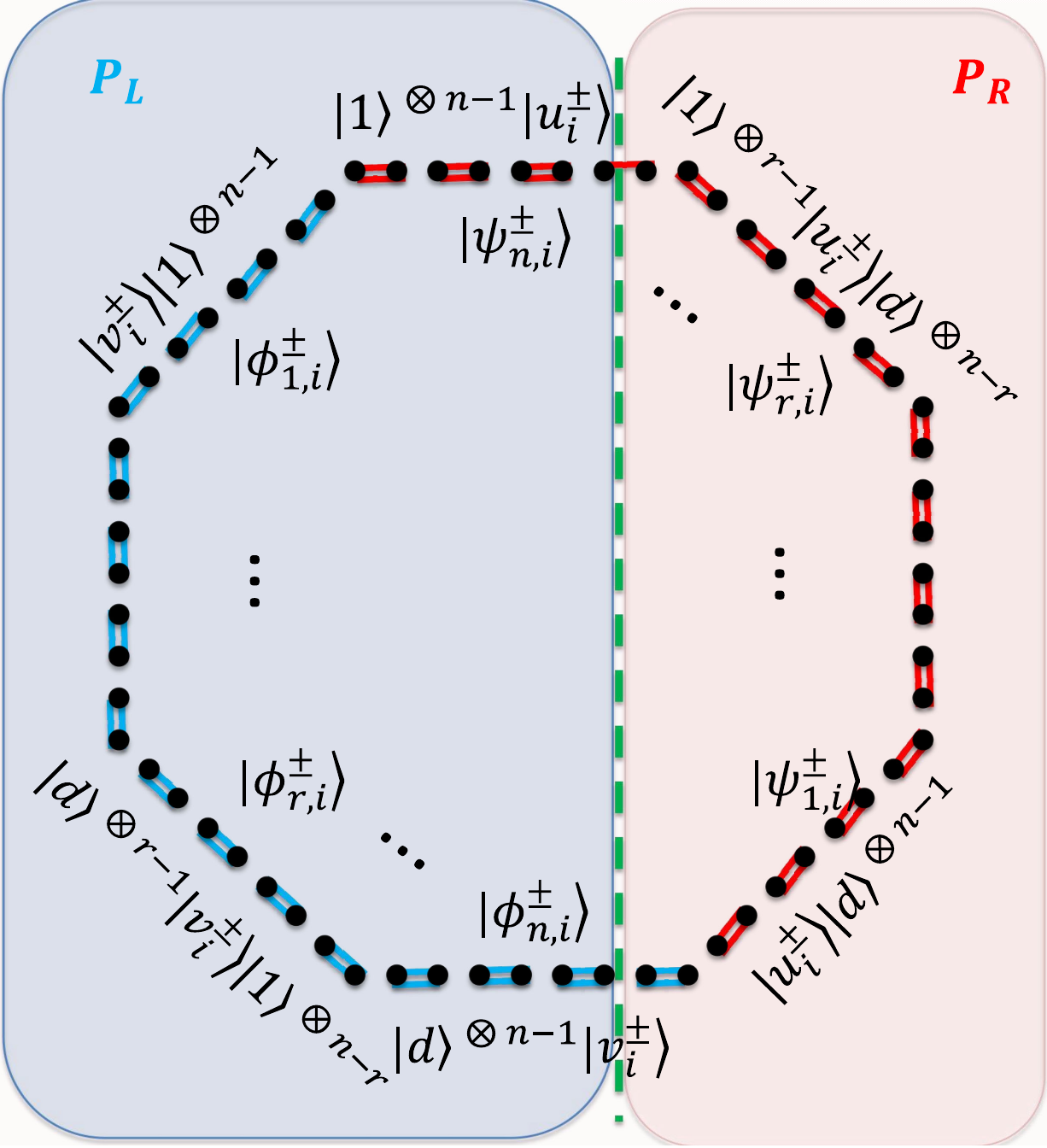}
    \vspace{1mm}
    
    {\small \textbf{(a)}}
\end{minipage}
\hfill
\begin{minipage}[t]{0.24\textwidth}
    \centering
    \includegraphics[width=\textwidth]{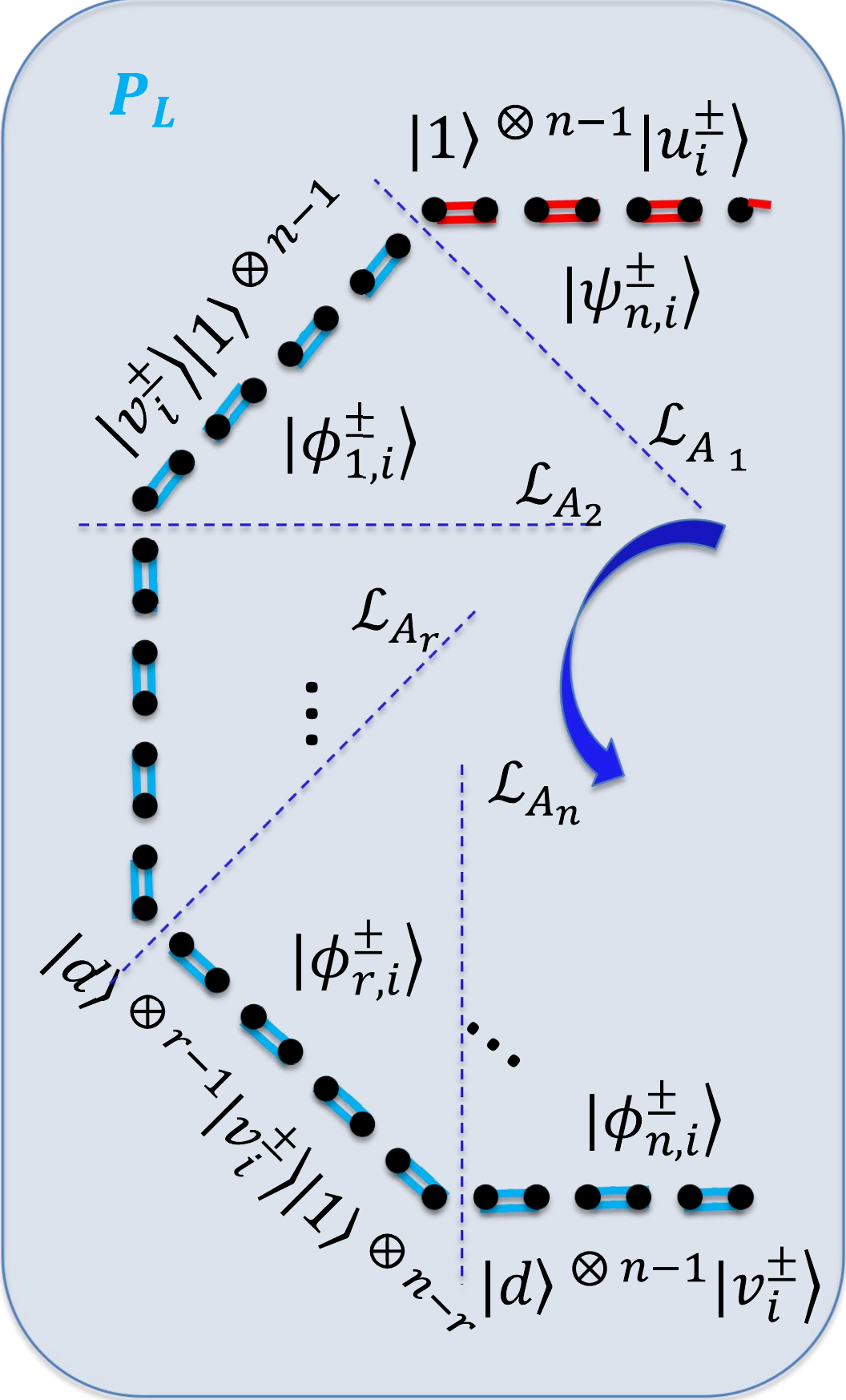}
    \vspace{1mm}
    
    {\small \textbf{(b)}}\label{Fig2b}
\end{minipage}
\hfill
\begin{minipage}[t]{0.24\textwidth}
    \centering
    \includegraphics[width=\textwidth]{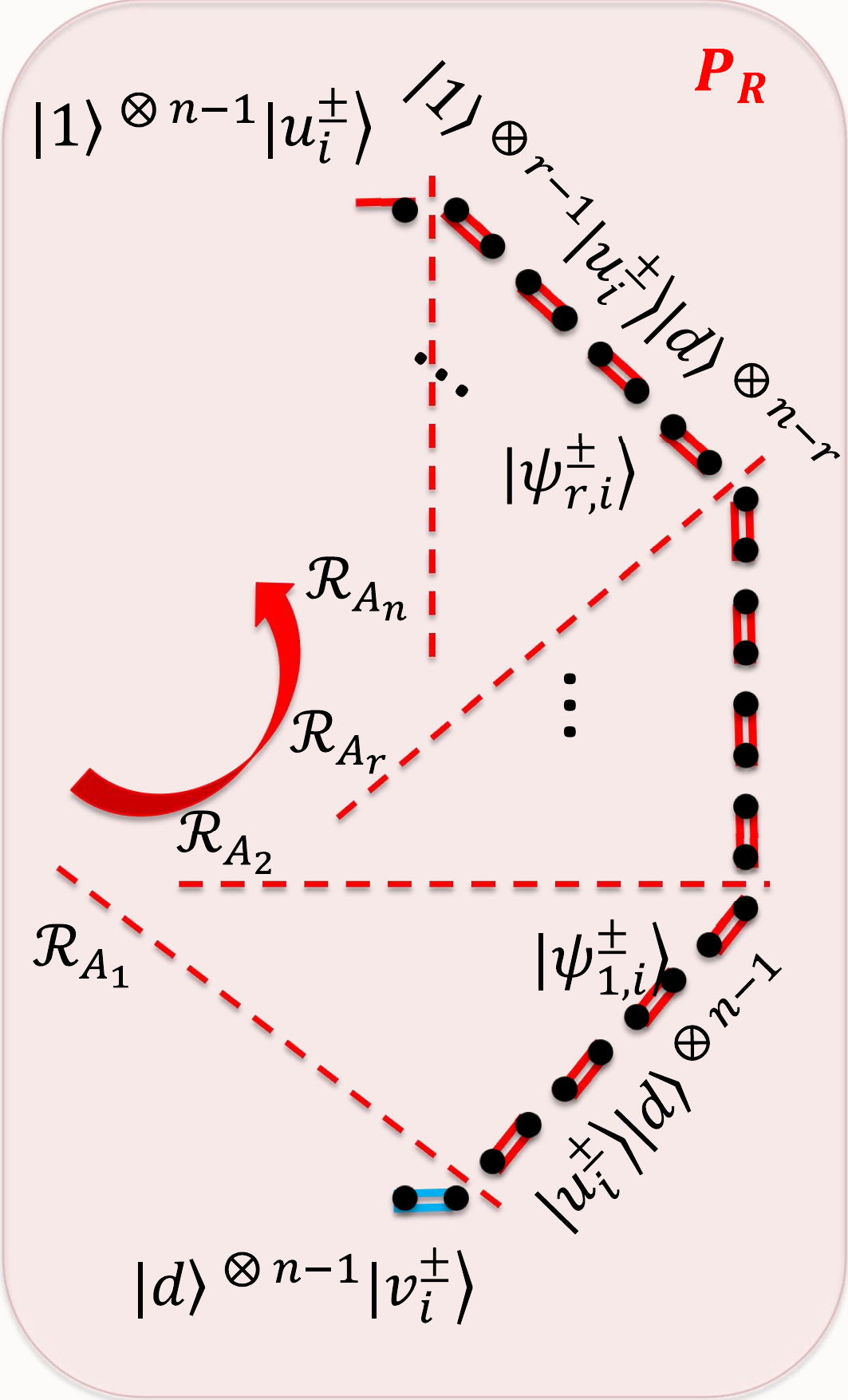}
    \vspace{1mm}
    
    {\small \textbf{(c)} }
\end{minipage}

\caption{(Color online) Schematic illustration of the construction and the LOCC protocols. 
	(a) Connectivity structure of the state families and the splitting of the set
	 $\mathcal S' $ induced by the first projective measurement on  $A_n $.
	(b) Sequential measurements on the  $P_L $ branch.
	(c) Sequential measurements on the  $P_R $ branch.
}
\label{fig:big-figure}
\end{figure*}

\begin{proof}
{ \bf Local indistinguishability:}
Fix a party  $A_m $, and let  $\{E_x=M_x^\dagger M_x\}_{x\in X} $ be an
orthogonality-preserving local POVM performed by  $A_m $. For one fixed outcome  $x $, write
 $$
E_x=(a_{pq})_{p,q=1}^d
 $$
in the computational basis of  $A_m $. We show that  $E_x $ is proportional to  $I_d $.

First, we will show that the off diagonal entries of $E_x$ vanish. Consider the states
 $\ket{\psi_{m,i}^{\epsilon}} $ and  $\ket{\psi_{m,l}^{\eta}} $, with  $i\ne l $ and
 $\epsilon,\eta\in\{+,-\} $, have identical local factors on all parties except  $A_m $. Hence
orthogonality preservation gives
 $$
\bra{u_i^\epsilon}E_x\ket{u_l^\eta}=0,
\qquad i\ne l,
 $$
for all four choices of signs. Solving these four equations gives
 $$
a_{pq}=0,
\qquad
p\in\{2i-1,2i\},\quad q\in\{2l-1,2l\},\quad i\ne l.
 $$
The same argument applied to
 $\ket{\phi_{m,i}^{\epsilon}} $ and  $\ket{\phi_{m,l}^{\eta}} $ gives
 $$
a_{pq}=0,
\qquad
p\in\{2i,2i+1\},\quad q\in\{2l,2l+1\},\quad i\ne l.
 $$

It remains to eliminate the entries involving the endpoints  $\ket1 $ and  $\ket d $.  If  $m>1 $, comparing
 $\ket{\psi_{m,i}^{\epsilon}} $ with  $\ket{\psi_{m-1,1}^{\eta}} $ gives
 $$
\bra{u_i^\epsilon}E_x\ket d=0,
 $$
and hence
 $$
a_{2i-1,d}=a_{2i,d}=0,
\qquad i=1,\ldots,k.
 $$
For  $m=1 $, the same endpoint constraint is obtained by going around the closed cycle and comparing
 $\ket{\psi_{1,i}^{\epsilon}} $ with  $\ket{\phi_{n,k}^{\eta}} $. Similarly, comparing adjacent
 $\phi $-families gives
 $$
\bra{v_j^\epsilon}E_x\ket1=0,
 $$
and therefore
 $$
a_{2j,1}=a_{2j+1,1}=0,
\qquad j=1,\ldots,k,
 $$
with the case  $m=1 $ again obtained by closing the cycle through the comparison of
 $\ket{\phi_{1,j}^{\epsilon}} $ with  $\ket{\psi_{n,1}^{\eta}} $. Together with Hermiticity, these
relations eliminate all off-diagonal entries, so
 $$
E_x=\operatorname{diag}(a_{11},a_{22},\ldots,a_{dd}).
 $$

Second, the diagonal entries are forced to be equal along the whole basis chain. Since  $E_x $ is
now diagonal, the orthogonality of the two states  $\ket{\psi_{m,i}^{+}} $ and
 $\ket{\psi_{m,i}^{-}} $ gives
 $$
0=\bra{u_i^+}E_x\ket{u_i^-}
=a_{2i-1,2i-1}-a_{2i,2i}.
 $$
Thus
 $$
a_{2i-1,2i-1}=a_{2i,2i},\qquad i=1,\ldots,k.
 $$
Likewise, the pair  $\ket{\phi_{m,i}^{+}} $,  $\ket{\phi_{m,i}^{-}} $ gives
 $$
a_{2i,2i}=a_{2i+1,2i+1},\qquad i=1,\ldots,k.
 $$
Hence
 $$
a_{11}=a_{22}=a_{33}=\cdots=a_{dd},
 $$
and therefore  $E_x\propto I_d $. Since the party and the outcome were arbitrary, every
orthogonality-preserving local POVM is trivial. Hence  $\mathcal S $ is 
locally indistinguishable.

\medskip
\noindent
\noindent{\bf Inclusion minimality:}
We now prove the inclusion-minimal property.  By the symmetries of the
construction, it is enough to delete one representative state.  Let
 $$
\ket{\psi_{n,t}^{-}}
=
\ket{1}^{\otimes(n-1)}\ket{u_t^-}_{A_n},
\qquad t\in\{1,\ldots,k\},
 $$
and put
 $$
\mathcal S'=\mathcal S\setminus\{\ket{\psi_{n,t}^{-}}\}.
 $$
Party  $A_n $ first performs the projective measurement
 $$
P_L=\sum_{q=1}^{2t-1}\ket q\bra q,
\qquad
P_R=\sum_{q=2t}^{d}\ket q\bra q.
 $$
As in the three-qutrit example, this measurement only splits the remaining
partner  $\ket{\psi_{n,t}^{+}} $.  The two postmeasurement endpoint states are
 $$
\ket{\gamma_L}
=
\ket{1}^{\otimes(n-1)}\ket{2t-1}_{A_n},
\qquad
\ket{\gamma_R}
=
\ket{1}^{\otimes(n-1)}\ket{2t}_{A_n}.
 $$
Thus the deleted state opens the closed chain into a  $P_L $ branch and a
 $P_R $ branch, as illustrated in Fig.~\ref{fig:big-figure} (a).

\medskip
\noindent
\textbf{The  $P_L $ branch, see in Fig. \ref{fig:big-figure} (b).}
After outcome  $P_L $, the possible states are
 $$
\{\ket{\phi_{r,j}^{\pm}}:1\le r\le n-1,\ 1\le j\le k\},
 $$
 $$
\{\ket{\phi_{n,j}^{\pm}}:1\le j\le t-1\},
\qquad
\{\ket{\psi_{n,i}^{\pm}}:1\le i\le t-1\},
 $$
together with  $\ket{\gamma_L} $.  This branch is distinguished by the same
left-endpoint screening as in the preceding example.

Party  $A_1 $ first tests whether its local state is  $\ket1 $:
 $$
\mathcal L_{A_1}=\{Q_1,Q_1^\perp\},
\qquad
Q_1=\ket1\bra1,\quad Q_1^\perp=I-Q_1 .
 $$
If the outcome is  $Q_1 $, the remaining candidates are
 $$
\ket{\gamma_L}
\quad\text{and}\quad
\{\ket{\psi_{n,i}^{\pm}}:1\le i\le t-1\}.
 $$
They are mutually orthogonal on  $A_n $, and  $A_n $ distinguishes them by a
projective measurement containing
 $$
\ket{2t-1},\qquad
\ket{u_i^\pm}\quad (i=1,\ldots,t-1).
 $$

If the outcome is  $Q_1^\perp $, the above candidates are excluded.  The
remaining candidates are  $\phi $-type states.  For
 $s=2,\ldots,n $, party  $A_s $ successively performs
 $$
\mathcal L_{A_s}=\{Q_s,Q_s^\perp\},
\qquad
Q_s=\ket1\bra1,\quad Q_s^\perp=I-Q_s .
 $$
When the first outcome  $Q_s $ occurs, the only remaining states are
 $$
\{\ket{\phi_{s-1,j}^{\pm}}:1\le j\le k\},
 $$
which party  $A_{s-1} $ distinguishes in a basis containing
 $\ket{v_j^\pm} $,  $j=1,\ldots,k $.  If all these tests give the orthogonal
outcome, then the only remaining states are
 $$
\{\ket{\phi_{n,j}^{\pm}}:1\le j\le t-1\},
 $$
which  $A_n $ distinguishes in a basis containing
 $\ket{v_j^\pm} $,  $j=1,\ldots,t-1 $.  Hence the  $P_L $ branch is perfectly
distinguishable by LOCC.

\medskip
\noindent
\textbf{The  $P_R $ branch, see in Fig. \ref{fig:big-figure} (c).}
After outcome  $P_R $, the possible states are
 $$
\{\ket{\psi_{r,i}^{\pm}}:1\le r\le n-1,\ 1\le i\le k\},
 $$
 $$
\{\ket{\psi_{n,i}^{\pm}}:t+1\le i\le k\},
\qquad
\{\ket{\phi_{n,j}^{\pm}}:t\le j\le k\},
 $$
together with  $\ket{\gamma_R} $.  This branch is distinguished by the
right-endpoint analogue of the preceding screening procedure.

Party  $A_1 $ first tests whether its local state is  $\ket d $:
 $$
\mathcal R_{A_1}=\{R_1,R_1^\perp\},
\qquad
R_1=\ket d\bra d,\quad R_1^\perp=I-R_1 .
 $$
If the outcome is  $R_1 $, the remaining candidates are
 $$
\{\ket{\phi_{n,j}^{\pm}}:t\le j\le k\},
 $$
which  $A_n $ distinguishes in a basis containing
 $\ket{v_j^\pm} $,  $j=t,\ldots,k $.

If the outcome is  $R_1^\perp $, these states are excluded.  For
 $s=1,\ldots,n-1 $, party  $A_{s+1} $ successively performs
 $$
\mathcal R_{A_{s+1}}=\{R_{s+1},R_{s+1}^\perp\},
\
R_{s+1}=\ket d\bra d, \  R_{s+1}^\perp=I-R_{s+1}.
 $$
When the first outcome  $R_{s+1} $ occurs, the only remaining states are
 $$
\{\ket{\psi_{s,i}^{\pm}}:1\le i\le k\},
 $$
which party  $A_s $ distinguishes in a basis containing
 $\ket{u_i^\pm} $,  $i=1,\ldots,k $.  If all these tests give the orthogonal
outcome, then the only remaining states are
 $$
\ket{\gamma_R}
\quad\text{and}\quad
\{\ket{\psi_{n,i}^{\pm}}:t+1\le i\le k\}.
 $$
They are mutually orthogonal on  $A_n $, so  $A_n $ distinguishes them in a
basis containing
 $$
\ket{2t},\qquad
\ket{u_i^\pm}\quad (i=t+1,\ldots,k).
 $$
Thus the  $P_R $ branch is also perfectly distinguishable by LOCC.

Combining the two branches, the parties can perfectly distinguish
 $$
\mathcal S\setminus\{\ket{\psi_{n,t}^{-}}\}
 $$
by LOCC.  Deleting  $\ket{\psi_{m,t}^{\pm}} $ for another party  $m $ is
handled by starting the same protocol at the corresponding cut.  Deleting a
 $\phi $-type state is reduced to the same argument by the basis reversal
 $\ket q\mapsto\ket{d+1-q} $, which exchanges the two families.  The change of
sign is handled by a phase flip on the corresponding two-dimensional block.
Hence deleting any single state from  $\mathcal S $ gives an LOCC
distinguishable set.

We have shown that  $\mathcal S $ is locally indistinguishable, while every
one-state-deleted subset is locally distinguishable.  By the finite-set
equivalence in the definition, every proper subset is locally distinguishable.
Therefore  $\mathcal S_{n,d} $ is inclusion-minimal locally indistinguishable.
\end{proof}

\bibliography{references}

\end{document}